\documentclass[12pt,a4paper]{article}

\usepackage{amsmath,amssymb,amsthm}
\usepackage[margin=2.5cm]{geometry}
\usepackage[colorlinks=true,linkcolor=blue,citecolor=blue,urlcolor=blue]{hyperref}
\usepackage{enumitem}
\usepackage{graphicx}

\theoremstyle{plain}
\newtheorem{theorem}{Theorem}[section]
\newtheorem{proposition}[theorem]{Proposition}

\theoremstyle{definition}

\theoremstyle{remark}
\newtheorem{remark}[theorem]{Remark}

\newcommand{\bx}{\mathbf{x}}

\title{Curvature Blindness from Polarity Breaks and\\
Orientation Channel Fragmentation in V1}

\author{Michael Menke}
\date{}

\begin{document}
\maketitle

\begin{abstract}
\noindent
We present a mathematical model of the curvature blindness
illusion in which sinusoids appear as
angular zigzags when drawn with alternating contrast polarity
against a gray background.
\\
\\
The model identifies two complementary mechanisms, both
operating in V1.  First, \emph{polarity channel separation}:
simple cells are selective for contrast polarity, and lateral
connections link only same polarity neurons; where the line
switches from darker than background to lighter than background
at each peak and trough, the encoding population changes and
the lateral chain is broken, segmenting the contour into
half-wavelength pieces.  Second, \emph{orientation channel
fragmentation}: at moderate contrast, the active orientation
window is narrow, and within each half-wavelength segment
no single orientation channel spans the full range of edge
normals; the inflection point at the center of each segment
anchors a locally straight percept.  Together, the two
mechanisms produce a zigzag: polarity breaks supply the
corners, and fragmentation straightens the segments between
them.
\\
\\
The model yields three necessary conditions for the illusion:
(i)~the curve must have contrast polarity reversals (to
create segment boundaries); (ii)~the contrast must be moderate
(to narrow orientation channels); and (iii)~the curve must have
inflection points between successive polarity reversals (to
provide locally straight anchors).  Condition~(iii) predicts
that the illusion should generalize to any polarity alternating
curve with sign-changing curvature (S-curves, damped
oscillations).
\end{abstract}

\section{Introduction}\label{sec:intro}

\subsection{The illusion}

The curvature blindness illusion, discovered by
Takahashi~\cite{takahashi17}, is a failure of curvature
perception.  In the standard stimulus, a sinusoidal line is
drawn against a gray background with its luminance alternating
between darker than background and lighter than background,
the switch occurring at each peak and trough.  The line appears
to have sharp corners at its peaks and troughs, as if it were
a zigzag.  The same line is correctly perceived as smooth at
high contrast against a white or black
background~\cite{takahashi17}.  The illusion requires that amplitude cannot be too large: large amplitude sinusoids are perceived as curved
even on a gray background~\cite{takahashi17}. Bertamini and
Kitaoka~\cite{bertamini20} suggested informally that
corner detecting mechanisms may be responsible, but did not
provide equations or testable conditions.

\begin{figure}[ht]
\centering
\includegraphics[width=0.6\textwidth]{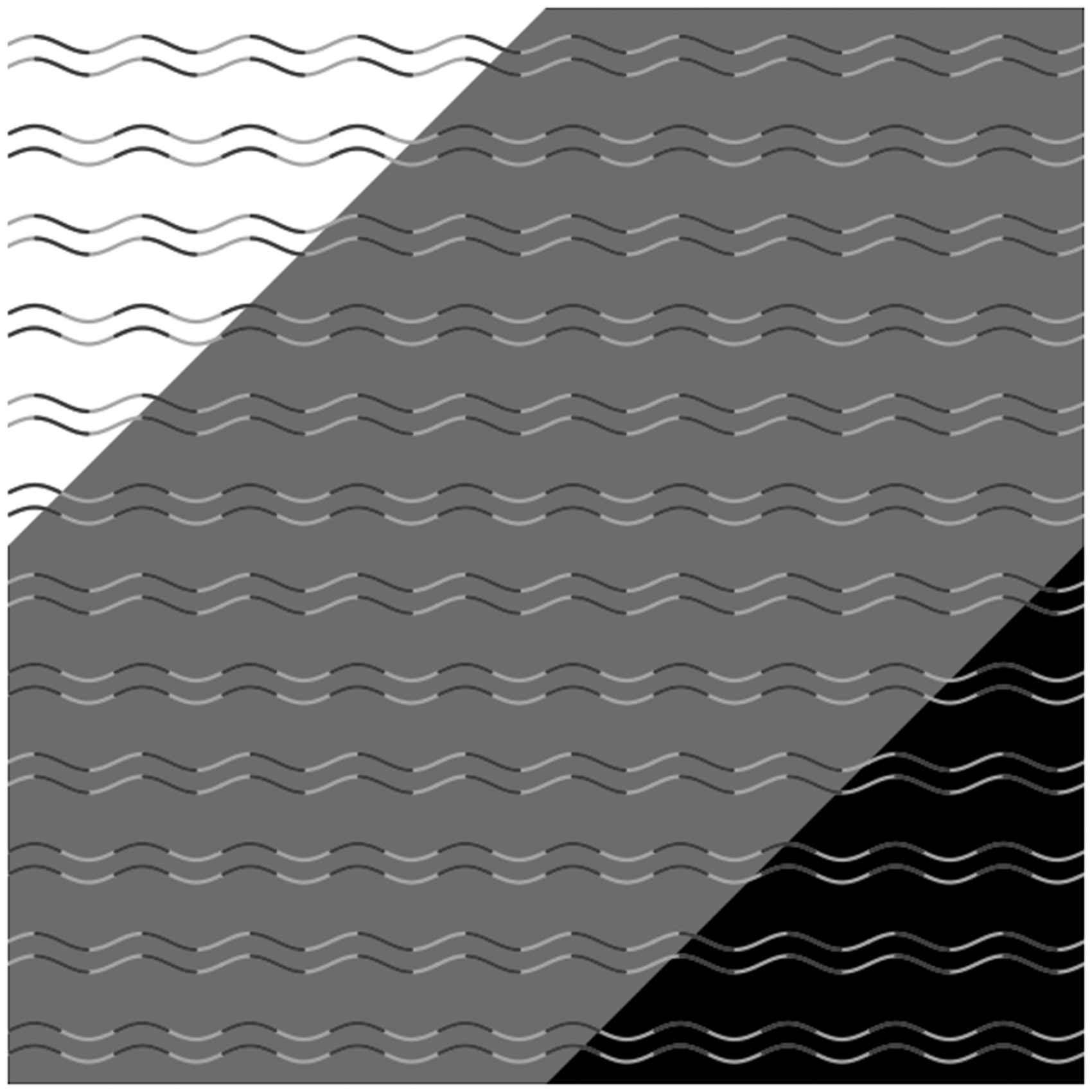}
\caption{The curvature blindness illusion (from
Takahashi~\cite{takahashi17}).  Sinusoidal curves with
alternating dark/light segments appear smooth against the
white and black backgrounds, but
appear as angular zigzags against the gray
background.}
\label{fig:illusion}
\end{figure}

\subsection{The main ideas}

We identify two V1 mechanisms that jointly produce
the illusion.
\\
\\
\emph{Polarity channel separation.}  V1 simple cells are
selective for contrast polarity: a cell that responds to a
dark to light edge does not respond to a light to dark edge at
the same orientation~\cite{hubel68}.  Lateral connections in V1
preferentially link neurons of matched polarity as well as
matched orientation~\cite{bosking97,field93}.
Where the stimulus line switches from
darker than background to lighter than background, the
encoding population changes entirely, and the lateral chain
that mediates contour integration is broken.  In the
Takahashi stimulus, these polarity switches occur precisely
at the peaks and troughs of the sinusoid, segmenting the
contour into half wavelength pieces.
\\
\\
\emph{Orientation channel fragmentation.}  Perceiving
curvature within a segment requires comparing edge orientations
across nearby positions, via lateral connections linking neurons
with matched preferred orientations~\cite{bosking97}.  At
moderate contrast, each edge point activates only a narrow range
of orientations (the \emph{active window}), and the window
center shifts as the edge normal rotates along the curve.  When
the total rotation within a half wavelength segment exceeds the
window width, no single orientation channel spans the full
segment and the representation fragments into short, nearly
straight pieces.  The inflection point at the center of each
segment, where curvature is zero and the tangent is steepest,
anchors the perceived straight line.
\\
\\
Together, the two mechanisms produce a zigzag: polarity breaks
at peaks and troughs supply the corners, and
orientation channel fragmentation straightens the segments.

\section{The Model}\label{sec:model}

\subsection{V1 orientation tuned edge detectors}

Each V1 hypercolumn at retinal position $\bx$ contains simple
cells tuned to every orientation
$\theta \in [0, \pi)$~\cite{hubel68}.  Each cell responds
primarily to luminance gradients aligned with its preferred
direction.  The key properties of simple cells relevant to our
model are:

\begin{enumerate}[label=(\roman*)]
\item \emph{Oriented linear filtering.}
  Simple cell receptive fields are well described by Gabor
  functions~\cite{jones87}.  For an odd symmetric
  (edge detecting) cell, the linear response to a luminance
  edge is proportional to the projection of the image gradient
  onto the cell's preferred direction.

\item \emph{Contrast polarity selectivity.}
  Simple cells are phase sensitive: an odd symmetric cell
  responds with opposite sign to opposite polarity
  edges~\cite{hubel68}.  After rectification, the population
  encoding a dark to light edge at orientation $\theta$ is
  distinct from the population encoding a light to dark edge
  at the same orientation.  We label these the $+$ and $-$
  polarity populations.

\item \emph{Gaussian orientation tuning.}
  The orientation tuning curve is approximately Gaussian, with
  a half-width at half-max (HWHM) of
  $\sim 20^\circ$~\cite{sclar82,ringach02}.  The corresponding
  Gaussian parameter is
  $\sigma_\theta = \mathrm{HWHM}/\!\sqrt{2\ln 2}
  \approx 17^\circ \approx 0.30$~rad.
  This bandwidth is approximately invariant with stimulus
  contrast~\cite{sclar82}.

\item \emph{Contrast detection threshold.}
  Each neuron has a minimum contrast below which it produces
  no detectable response above its spontaneous
  firing~\cite{hubel68}.

\item \emph{Divisive normalization.}
  The contrast response is governed by divisive
  normalization~\cite{heeger92,carandini12}, well described by
  the Naka--Rushton equation $R(c) = R_{\max}\,c^n/(\sigma_c^n + c^n)$
  with exponent $n \approx 2$ and semi-saturation constant
  $\sigma_c \approx 0.15$ (Michelson contrast)~\cite{carandini12,albrecht82}.

\item \emph{Lateral connections.}
  Long range horizontal connections link neurons at nearby
  spatial positions whose preferred orientations match the
  axis connecting them~\cite{bosking97}.  Psychophysical
  evidence indicates that contour integration is
  polarity specific~\cite{field93}, consistent with lateral
  connections that preferentially link same polarity neurons.
\end{enumerate}

\subsection{The linear response at an edge}\label{sec:linear_response}

A \emph{step edge} is a sharp boundary between two regions of
different luminance. For instance, the border between a dark
line and a lighter background.  At every point along such a
boundary, two quantities characterize the edge: the
\emph{normal direction} $\theta_n$ (the direction perpendicular
to the boundary, pointing from the darker side to the lighter
side), and the \emph{Michelson contrast}
$c = (I_{\max} - I_{\min})/(I_{\max} + I_{\min})$, where
$I_{\max}$ and $I_{\min}$ are the luminances on the two sides.
\\
\\
The edge also has a \emph{polarity} $p \in \{+, -\}$,
determined by which side is brighter.  Consider a neuron whose
preferred direction $\hat{e}_\theta$ points in the direction of
increasing luminance across the edge.  Its linear response is
positive: the luminance gradient is aligned with the neuron's
preferred axis.  A neuron whose preferred direction points the
opposite way produces a
negative linear response of equal magnitude.  We call the first
neuron a $+$ polarity neuron (for this edge) and the second a
$-$ polarity neuron.
\\
\\
Combining the cosine projection (property~(i)) with the
Gaussian tuning envelope (property~(iii)), the linear response
of a $+$ polarity neuron tuned to orientation $\theta$ at an
edge with normal $\theta_n$ and contrast $c$ is
\begin{equation}\label{eq:linear_response}
  L(\theta;\, \theta_n, c)
  = c \;\cdot\; h(\theta - \theta_n),
\end{equation}
where
\begin{equation}\label{eq:h_def}
  h(\delta) \;=\;
  \cos\delta \;\cdot\;
  \exp\!\Bigl(-\frac{\delta^2}{2\sigma_\theta^2}\Bigr)
\end{equation}
is the \emph{orientation response profile}, with $\delta = \theta - \theta_n$ the angular offset between the
neuron's preferred orientation and the edge normal. The profile is the product of the
geometric cosine projection and the Gaussian tuning envelope.
The response is maximal when $\theta = \theta_n$ (neuron
aligned with the gradient) and falls to zero when $\theta$
deviates far from $\theta_n$.
\\
\\
A $-$ polarity neuron at the same edge produces the linear
response $-L(\theta;\, \theta_n, c)$: negative, because the
luminance gradient opposes its preferred direction.  However,
V1 simple cells undergo half-wave rectification before their
signals are transmitted: the firing rate of a neuron cannot go
below its spontaneous baseline, so negative linear responses
are clipped to zero~\cite{hubel68}.  After rectification, only
the $+$ polarity population produces a nonzero response to this
edge; the $-$ polarity population is silent.  At a nearby edge
of opposite polarity (e.g., the other side of a line, where
the luminance gradient reverses), the roles are exchanged: the
$-$ population responds and the $+$ population is silent.  This
is the physiological basis for the polarity selectivity stated
in property~(ii).

\subsection{Divisive normalization}\label{sec:normalization}

Following the standard
model~\cite{heeger92,carandini12}, the firing rate of a neuron
tuned to orientation $\theta$ (with matching polarity) at an
edge with normal $\theta_n$ and Michelson contrast $c$ is
\begin{equation}\label{eq:normalization}
  R(\theta;\, \theta_n, c)
  = R_{\max}\,\frac{
    \bigl[\,c \;\cdot\; h(\theta - \theta_n)\,\bigr]^n
  }{
    \sigma_c^{\,n} \;+\; c^{\,n}\,\mathcal{C}_n
  },
\end{equation}
where the normalization pool energy
\begin{equation}\label{eq:Cn}
  \mathcal{C}_n
  = \int_{-\pi/2}^{\pi/2}
    \bigl[\,h(\varphi)\,\bigr]^n \, d\varphi
\end{equation}
sums the $n$th-power responses across all orientation channels
driven by the same edge.  With $n = 2$ and
$\sigma_\theta \approx 0.30$~rad, numerical integration gives
$\mathcal{C}_2 \approx 0.50$.
\\
\\
$R_{\max}$ is a gain parameter; the actual peak firing rate
at optimal orientation and high contrast is
$R_{\max}/\mathcal{C}_n$.  With $R_{\max} \approx 50$, this
gives $\sim 100$~sp/s, consistent with V1
data~\cite{hubel68}.

\subsection{The active orientation window}\label{sec:active_window}

A neuron is \emph{active} when
$R(\theta) \geq r_{\mathrm{noise}}$, where
$r_{\mathrm{noise}}$ is the minimum detectable elevation
above baseline.  Define the relative noise threshold
$\rho = r_{\mathrm{noise}} / R_{\max} \approx 0.10$.
Substituting~\eqref{eq:normalization} into the activity
condition and simplifying:
\begin{equation}\label{eq:active_condition}
  h(\theta - \theta_n)
  \;\geq\;
  \Bigl[\,
    \rho\,\Bigl(\frac{\sigma_c^{\,n}}{c^n}
    + \mathcal{C}_n\Bigr)
  \,\Bigr]^{1/n}
  \;\equiv\;
  \tau_{\mathrm{eff}}(c).
\end{equation}

\begin{proposition}[Active orientation window]
\label{prop:active_window}
At an edge with normal $\theta_n$ and Michelson contrast $c$,
a neuron of matching polarity at orientation $\theta$ is
active if and only if $|\theta - \theta_n| \leq \alpha(c)$,
where
\begin{equation}\label{eq:alpha_full}
  h\bigl(\alpha(c)\bigr)
  = \tau_{\mathrm{eff}}(c),
\end{equation}
provided $\tau_{\mathrm{eff}}(c) \leq 1$.  If
$\tau_{\mathrm{eff}}(c) > 1$, no neuron is active and the
edge is invisible.  A neuron of the opposite polarity is never
active at this edge.
\end{proposition}

\begin{proof}
Since $h$ is strictly decreasing on $[0, \pi/2)$ with
$h(0) = 1$ and $h(\pi/2) = 0$,
Equation~\eqref{eq:alpha_full} has a unique solution.
The neuron is active iff $h(\theta - \theta_n) \geq
\tau_{\mathrm{eff}}$, i.e., $|\theta - \theta_n| \leq \alpha$.
The opposite-polarity statement follows from rectification.
\end{proof}
\noindent
The effective threshold $\tau_{\mathrm{eff}}(c)$ controls how
wide the active window is: a lower threshold means more
orientations are active.  To understand its behavior, note
that $\tau_{\mathrm{eff}}(c)$ consists of two terms inside
the bracket:
\[
  \tau_{\mathrm{eff}}(c)
  = \bigl[\,\rho\,\sigma_c^{\,n}/c^n
    + \rho\,\mathcal{C}_n\,\bigr]^{1/n}.
\]
The first term, $\rho\,\sigma_c^n / c^n$, comes from the
semi-saturation constant in the normalization denominator: it
dominates at low contrast and falls as contrast increases.
The second term, $\rho\,\mathcal{C}_n$, comes from the
normalization pool and is independent of contrast: it sets a
floor that $\tau_{\mathrm{eff}}$ can never drop below, no
matter how high the contrast.  These two terms produce three
regimes:

\begin{enumerate}[label=(\roman*)]
\item \emph{Below visibility} ($c < c_{\mathrm{vis}}$).
  At very low contrast, the first term $\rho\,\sigma_c^n/c^n$
  is large, making $\tau_{\mathrm{eff}} > 1$.  Since
  $h(0) = 1$ is the maximum of $h$, no orientation satisfies
  $h(\delta) \geq \tau_{\mathrm{eff}}$, and the edge is
  invisible.  The visibility threshold $c_{\mathrm{vis}}$ is
  the contrast at which $\tau_{\mathrm{eff}} = 1$ exactly,
  found by setting $\rho\,(\sigma_c^n/c^n + \mathcal{C}_n) = 1$
  and solving for $c$:
  \begin{equation}\label{eq:c_vis}
    c_{\mathrm{vis}}
    = \sigma_c\,
    \Bigl(\frac{\rho}{1 - \rho\,\mathcal{C}_n}\Bigr)^{1/n}
    \approx 0.049
  \end{equation}
  ($\sim 5\%$ Michelson) with the standard parameters
  ($\sigma_c = 0.15$, $\rho = 0.10$, $n = 2$,
  $\mathcal{C}_2 \approx 0.50$).  Below this contrast, no V1
  neuron fires above its noise floor in response to the edge.

\item \emph{Moderate contrast}
  ($c$ somewhat above $c_{\mathrm{vis}}$).  The first term is
  still substantial, so $\tau_{\mathrm{eff}}$ is close to~$1$
  and the active window $\alpha$ is narrow.  In this regime,
  $h(\delta) \approx \cos\delta$ for small $\delta$ (the
  Gaussian factor is nearly flat near $\delta = 0$), and the
  window half-width is approximately
  $\alpha \approx \arccos(\tau_{\mathrm{eff}})$. This is the
  regime in which orientation channel fragmentation is most
  severe: the window is too narrow to span the tangent
  variation of a sinusoidal edge.

\item \emph{High contrast} ($c \gg \sigma_c$).
  The first term $\rho\,\sigma_c^n/c^n$ becomes negligible,
  and $\tau_{\mathrm{eff}}$ approaches the floor set by the
  second term:
  \begin{equation}\label{eq:tau_inf}
    \tau_\infty
    = \bigl(\rho\,\mathcal{C}_n\bigr)^{1/n}
    \approx 0.224.
  \end{equation}
  The active window therefore approaches
  \begin{equation}\label{eq:alpha_inf}
    \alpha_\infty
    = h^{-1}(\tau_\infty)
    \approx 28^\circ.
  \end{equation}
  No matter how high the contrast, the active window cannot
  exceed $2\alpha_\infty \approx 56^\circ$.  This ceiling
  arises because divisive normalization compresses the peak
  response: at high contrast, all neurons fire proportionally
  harder, but the normalization pool grows in step, preventing
  any neuron from rising far enough above threshold to activate
  neurons at distant orientations.  The Gaussian tuning
  envelope reinforces this ceiling by ensuring that neurons far
  from the edge normal produce negligible responses regardless
  of contrast.
\end{enumerate}

\subsection{Lateral connections: orientation and polarity}
\label{sec:lateral}

V1 lateral connections (property~(vi)) link neurons at nearby
positions with matched preferred orientations and matched
contrast polarity~\cite{bosking97,field93}.  A lateral
connection can therefore carry signal between two edge
positions $\bx_1$ and $\bx_2$ only if:
\begin{enumerate}[label=(\alph*)]
\item the edges at $\bx_1$ and $\bx_2$ have the \emph{same
  contrast polarity}; and
\item there exists an orientation $\theta_0$ within the
  active windows at \emph{both} positions.
\end{enumerate}
\noindent
Condition~(a) fails at a polarity reversal: the encoding
populations on the two sides of the reversal are entirely
disjoint, regardless of orientation.  Condition~(b) fails
when the active windows
$[\theta_{n,1} - \alpha, \theta_{n,1} + \alpha]$ and
$[\theta_{n,2} - \alpha, \theta_{n,2} + \alpha]$ are
disjoint, i.e., when $|\theta_{n,1} - \theta_{n,2}| > 2\alpha$.
\\
\\
The two conditions correspond to the two mechanisms of the
model: polarity channel separation (a) and orientation channel
fragmentation (b).

\section{Curvature Blindness}\label{sec:curvature_blindness}

\subsection{The Takahashi stimulus}

The standard curvature blindness stimulus consists of a
sinusoidal curve $y = A\sin(2\pi x/\lambda)$ drawn against a
uniform gray background of luminance $I_{\mathrm{bg}}$.  The
luminance of the line alternates between
$I_{\mathrm{dark}} < I_{\mathrm{bg}}$ and
$I_{\mathrm{light}} > I_{\mathrm{bg}}$, with the switch
occurring at each peak ($y = A$) and trough ($y = -A$).  Each
half wavelength segment, from a peak to the next trough  or
from a trough to the next peak, has a single, uniform
contrast polarity against the background.

\subsection{Polarity channel separation at peaks and troughs}

\begin{proposition}[Polarity segmentation]
\label{prop:polarity}
At each peak and trough of the Takahashi stimulus, the
contrast polarity reverses.  By condition~(a) of
Section~\ref{sec:lateral}, lateral connections cannot bridge
the polarity boundary.  The contour representation is therefore
segmented into half wavelength pieces, each encoded by a
single polarity neural population.
\end{proposition}

\begin{proof}
The segment from peak to trough (say) has luminance
$I_{\mathrm{dark}} < I_{\mathrm{bg}}$: its edges are
encoded by the $-$ polarity population.  The next segment
(trough to peak) has luminance
$I_{\mathrm{light}} > I_{\mathrm{bg}}$: its edges are
encoded by the $+$ population.  These populations are
disjoint (property~(ii)), so no lateral connection spans the
boundary.
\end{proof}

\begin{remark}[Why the illusion requires a gray background]
\label{rem:gray_background}
On a \emph{white} background
($I_{\mathrm{bg}} \gg I_{\mathrm{light}} > I_{\mathrm{dark}}$),
both line luminances are darker than the background.
Both segments therefore have the \emph{same} contrast polarity
($-$) relative to the background.  No polarity reversal occurs
at peaks and troughs, the lateral chain is unbroken, and the
contour is perceived as smooth.  On a \emph{black} background,
both luminances are lighter than the background, both have
polarity~$+$, and the same argument applies.
\\
\\
On a \emph{gray} background where
$I_{\mathrm{dark}} < I_{\mathrm{bg}} < I_{\mathrm{light}}$,
the dark segments have polarity~$-$ and the light segments
have polarity~$+$.  The polarity reverses at every peak and
trough, the lateral chain is broken, and the segmentation of
Proposition~\ref{prop:polarity} is in effect.  The background
dependence of the illusion is thus entirely explained by
whether the background luminance falls between or outside the
two line luminances.
\end{remark}

\subsection{Geometry of a half wavelength segment}

Each half wavelength segment runs from a peak to the next
trough (or vice versa) and has uniform contrast polarity.
Within such a segment, the geometry is as follows.

\begin{proposition}[Half segment geometry]
\label{prop:half_segment}
Consider the segment of $y = A\sin(2\pi x/\lambda)$ from a
peak at $x_0 = \lambda/4$ to the next trough at
$x_1 = 3\lambda/4$.
\begin{enumerate}[label=(\roman*)]
\item The tangent angle $\theta_{\mathrm{tan}}(x)
  = \arctan(y'(x))$ varies from $0$ at the peak to
  $-\theta_{\max}$ at the midpoint
  $x_m = \lambda/2$ and back to $0$ at the trough,
  where $\theta_{\max} = \arctan(2\pi A/\lambda)$.

\item The midpoint $x_m$ is an \emph{inflection point}:
  $y''(x_m) = 0$.  Here the curvature is zero and
  $|\theta_{\mathrm{tan}}|$ is maximal; the curve is
  locally straightest.

\item The edge normal
  $\theta_n(x) = \theta_{\mathrm{tan}}(x) + \pi/2$
  varies within
  $[\pi/2 - \theta_{\max},\; \pi/2]$
  over this segment (and symmetrically within
  $[\pi/2,\; \pi/2 + \theta_{\max}]$ over the adjacent
  trough to peak segment).

\item The Michelson contrast $c$ between the line and the
  background is uniform along the segment (single polarity,
  uniform luminance).  Therefore the active window has the
  same width $2\alpha(c)$ everywhere within the segment.
\end{enumerate}
\end{proposition}

\begin{proof}
The derivative is $y'(x) = 2\pi A\cos(2\pi x/\lambda)/\lambda$.
At $x_0 = \lambda/4$: $\cos(\pi/2) = 0$, so $\theta_{\mathrm{tan}} = 0$.
At $x_m = \lambda/2$: $\cos(\pi) = -1$, so
$\theta_{\mathrm{tan}} = \arctan(-2\pi A/\lambda) = -\theta_{\max}$.
At $x_1 = 3\lambda/4$: $\cos(3\pi/2) = 0$, so $\theta_{\mathrm{tan}} = 0$.
The second derivative $y''(x_m) = -(2\pi/\lambda)^2 A\sin(\pi) = 0$
confirms the inflection point.  Parts~(iii)--(iv) follow directly.
\end{proof}

\subsection{Orientation channel fragmentation within a segment}

\begin{proposition}[Within segment fragmentation]
\label{prop:fragmentation}
When $\theta_{\max} > 2\alpha(c)$, the active orientation
windows at the endpoints (peak/trough, where
$\theta_n = \pi/2$) and at the inflection point (where
$\theta_n = \pi/2 - \theta_{\max}$) are disjoint.  No single
orientation channel spans the full half wavelength segment.
\end{proposition}

\begin{proof}
The window at the peak is centered at $\pi/2$ with half-width
$\alpha$; the window at the inflection is centered at
$\pi/2 - \theta_{\max}$ with half-width $\alpha$.  These are
disjoint when $\theta_{\max} > 2\alpha$.
\\
\\
Every point within the segment remains visible: $\theta_n(x)$
varies continuously, and at each $x$ the orientation
$\theta = \theta_n(x)$ lies within the local active window.
The gap is in the orientation channel, not in space.
\end{proof}

\subsection{The role of inflection points}

\begin{proposition}[Inflection points anchor straight line perception]
\label{prop:inflection}
When fragmentation holds ($\theta_{\max} > 2\alpha$), the
orientation channel centered near
$\theta_0 = \pi/2 - \theta_{\max}$ sees a fragment of the
curve surrounding the inflection point.  This fragment is the
straightest part of the segment: it is centered on the point
of zero curvature, and the tangent angle varies by at most
$2\alpha$ across it.
\end{proposition}

\begin{proof}
The fragment visible at orientation $\theta_0$ consists of
positions $x$ where $|\theta_n(x) - \theta_0| \leq \alpha$.
At the inflection point $x_m$,
$\theta_n(x_m) = \pi/2 - \theta_{\max}$, which equals
$\theta_0$; so $x_m$ is in the center of this fragment.
The tangent angle within the fragment varies by at most
$2\alpha$, and since $\kappa(x_m) = 0$ the curve is locally
straight at the center of the fragment.
\end{proof}

\begin{remark}[Why inflection points are necessary]
\label{rem:inflection_necessary}
Inflection points play a specific structural role: they
provide the locally straight anchors that the visual system
interprets as the ``segments'' of a zigzag.  Without inflection
points, for example along a circular arc where curvature
never changes sign, fragmentation still breaks the orientation
representation into pieces, but every piece is a curved arc
of identical shape.  There is no locally straight region to
anchor a straight segment perception.  The model therefore predicts
that the illusion requires \emph{sign changing curvature}
(i.e., inflection points between successive polarity reversals),
not merely polarity switching and moderate contrast.
\end{remark}

\subsection{The zigzag perception}

\begin{proposition}[Piecewise linear perception]
\label{prop:zigzag}
When both polarity segmentation
(Proposition~\ref{prop:polarity}) and within segment
fragmentation (Proposition~\ref{prop:fragmentation}) hold,
the visual system perceives a zigzag.
\end{proposition}

\begin{proof}
Polarity segmentation breaks the contour at every peak and
trough.  Within each half wavelength segment, fragmentation
prevents any single orientation channel from tracking the full
variation from $\theta_n = \pi/2$ at the endpoints to
$\theta_n = \pi/2 - \theta_{\max}$ (or
$\pi/2 + \theta_{\max}$ for the adjacent segment) at the
inflection.
The inflection centered fragment anchors a straight percept
(Proposition~\ref{prop:inflection}).  The fragments near the
endpoints see a nearly horizontal edge
($\theta_n \approx \pi/2$), consistent with a corner at the
peak or trough.
\\
\\
The visual system therefore receives: a polarity break
(segment boundary) at each peak and trough, flanked by
approximately straight segments passing through the inflection
points.  The simplest contour consistent with these constraints
is a triangular wave i.e a zigzag with corners at the peaks
$(y = A)$ and troughs $(y = -A)$, connected by straight
segments with slope $\pm 4A/\lambda$ and corner angle
$2\arctan(4A/\lambda) \approx 8A/\lambda$ for small amplitude.
\end{proof}

\subsection{How curvature and amplitude affect the illusion}

\begin{remark}\label{rem:large_amplitude}
For large amplitude ($A \sim \lambda$), $\theta_{\max}$
is large and the fragmentation condition is easily satisfied,
yet the illusion vanishes.  The resolution involves the
\emph{spatial length} of the fragments that each orientation
channel sees.
\\
\\
A neuron tuned to orientation $\theta_0$ sees all curve
positions $x$ where $|\theta_n(x) - \theta_0| \leq \alpha$.
The edge normal $\theta_n$ changes along the curve at a rate
determined by the curvature: specifically,
$d\theta_n/ds = \kappa$, where $s$ is arc length and $\kappa$
is the signed curvature.  If $\kappa$ is approximately
constant over the fragment, then the edge normal rotates by
$|\kappa| \cdot \Delta s$ over a fragment of arc length
$\Delta s$.  The fragment spans the full active window
$2\alpha$ in orientation, so
$|\kappa| \cdot \Delta s \approx 2\alpha$, giving
\begin{equation}\label{eq:fragment_length}
  \Delta s \;\approx\; \frac{2\alpha}{|\kappa|}.
\end{equation}
At large amplitude ($A \sim \lambda$), the curvature at the
peaks and troughs is large
($|\kappa| \approx (2\pi/\lambda)^2 A$), so each fragment
is spatially short.  The piecewise-linear approximation then
consists of many tiny segments, indistinguishable from a smooth
curve at the spatial resolution of the visual system.
\\
\\
However, curvature that is too \emph{small} also undermines
the zigzag, for the opposite reason.  When $|\kappa|$ is
small, Equation~\eqref{eq:fragment_length} gives a long
fragment, and the tangent angle varies by the full $2\alpha$
across it.  Although $2\alpha$ is small in absolute terms
(at moderate contrast, $2\alpha \lesssim 56^\circ$), the
curvature within the fragment is nonzero, the fragment is an
arc, not a line.  If the fragment is long enough, this
residual curvature is perceptible: the visual system can
detect that the fragment bends, even though it cannot compare
orientations \emph{across} fragments.  The straight segment
perception requires fragments that are both long enough to be
individually visible and short enough that their internal
curvature is below the detection threshold for bending.
\\
\\
The zigzag illusion therefore occupies an intermediate
regime: curvature must be high enough that each fragment is
short enough to appear straight, but low enough that each
fragment is long enough to be perceived as a distinct segment
rather than an indistinguishable piece of a smooth curve.
\end{remark}

\section{Predictions and Limitations}\label{sec:discussion}

\subsection{Predictions}

The model yields three necessary conditions and three
specific predictions.
\\
\\
\emph{Necessary conditions:}
\\
\\
The illusion requires:
(a)~contrast polarity reversals along the curve (to create
segment boundaries at peaks/troughs); (b)~moderate contrast
against the background (to narrow the active orientation
window); and (c)~inflection points between successive polarity
reversals (to anchor straight segment perception within each
half wavelength piece).
\\
\\
\emph{Predictions:}
\\
\\
(i)~\emph{Generalization to other curves.}  Any
polarity alternating curve with sign changing curvature should
show the illusion: damped sinusoids, S-curves, and
higher order oscillations.  Curves without inflection points
between polarity switches (arcs of circles or ellipses with
polarity alternation) should \emph{not} produce a zigzag
percept, though they may appear as disconnected arcs.
\\
\\
(ii)~\emph{Contrast window.}  Within each half wavelength
segment, the fragmentation condition
$\theta_{\max} > 2\alpha(c)$ defines a contrast range
$[c_{\mathrm{vis}}, c_{\mathrm{crit}}]$ in which the illusion
occurs.  The critical contrast $c_{\mathrm{crit}}$ depends on
$A/\lambda$ and exists whenever
$\theta_{\max} < 2\alpha_\infty \approx 56^\circ$
($A/\lambda \lesssim 0.24$).
\\
\\
(iii)~\emph{Amplitude dependent contrast window.}  The upper
critical contrast $c_{\mathrm{crit}}$ should increase with
$A/\lambda$: larger amplitude sinusoids should exhibit the
illusion over a wider contrast range.

\subsection{Limitations}

The model operates at the level of V1 and does not account for
curvature selective neurons in V2/V4 that might partially
compensate for orientation channel fragmentation.
\\
\\
The claim that lateral connections are polarity specific is
supported by the psychophysics of contour
integration~\cite{field93} and by the physiology of horizontal
connections~\cite{bosking97}, but the degree of polarity
selectivity in lateral connections has not been measured as
precisely as the orientation selectivity.  If lateral
connections have partial polarity tolerance, the polarity break
at peaks/troughs would be leaky rather than absolute, and the
illusion strength would degrade gradually rather than switching
sharply.
\\
\\
The model uses a single achromatic contrast channel; extending
to chromatic contrast (which also produces the
illusion~\cite{takahashi17}) would require replacing the scalar
Michelson contrast with a vector-valued color contrast, but the
mechanism is unchanged.
\\
\\
The model does not specify the mechanism that
converts fragmented orientation representations into a
conscious zigzag perception.  We have argued that the zigzag is
the simplest contour consistent with the available information
(polarity breaks at peaks/troughs, approximately straight
fragments at inflection points), but a complete model would
need to formalize this inference step.

\begin{figure}[ht!]
\centering
\includegraphics[width=0.75\textwidth]{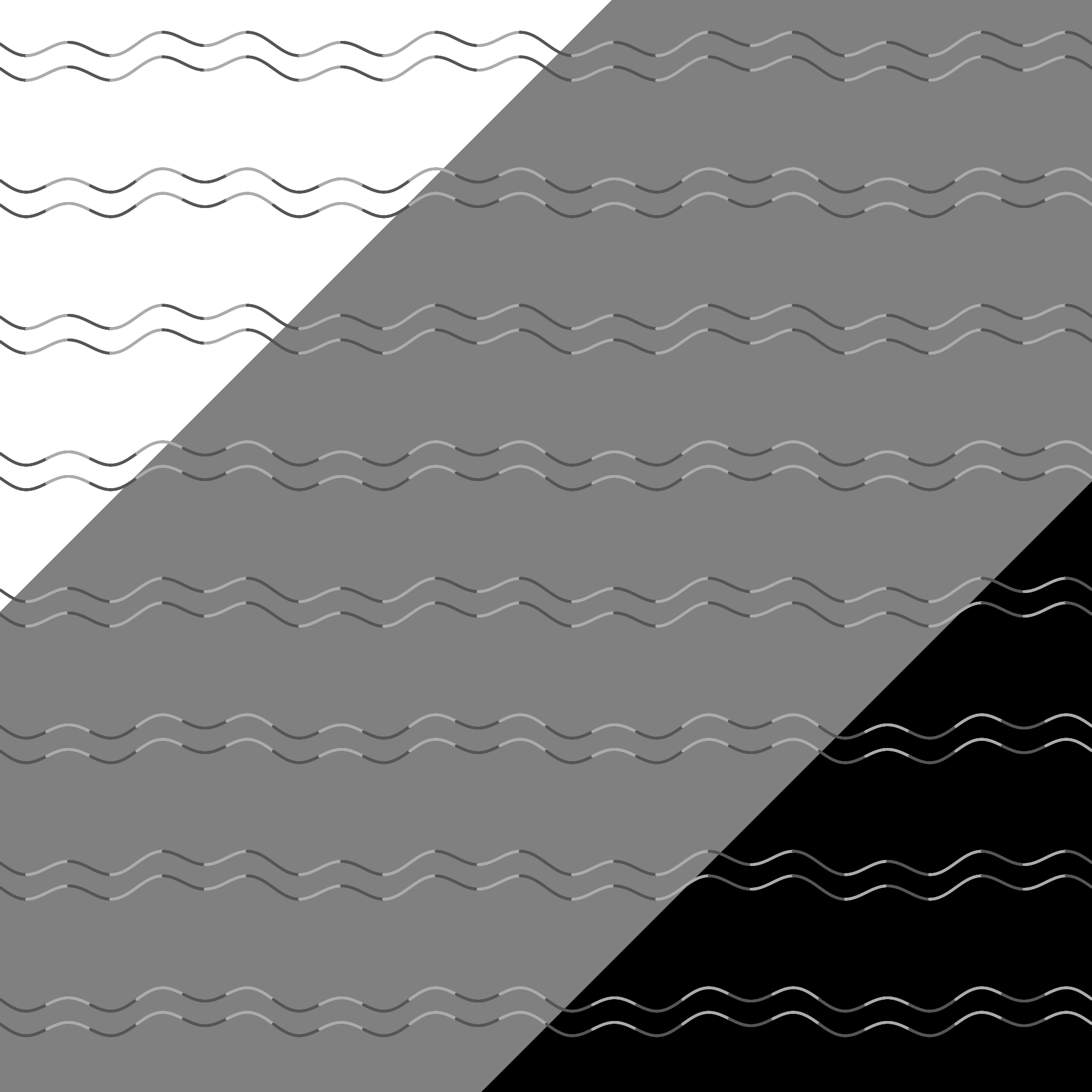}
\caption{Curvature blindness generalized to a non-sinusoidal
curve. \\The  waveform is $y = 4\sin(2\pi x/\lambda) + 5\sin(6\pi x/\lambda).$ The longer segments between extrema still read as curved while the shorter segments are straight. This is consistent with the model as the longer segments have lower average curvature and hence a larger visible window allowing residual curvature to be seen.
}
\label{fig:twofreq}
\end{figure}

\end{document}